\documentclass{llncs}

\usepackage{amsmath}
\usepackage{amsfonts}
\usepackage{amssymb}
\usepackage{todonotes}
\usepackage[all]{xy}
\usepackage{wrapfig}

\newcommand{\ACC}{\mathsf{ACC}}
\newcommand{\body}{\mathsf{body}}
\newcommand{\BS}{\mathsf{BS}}
\newcommand{\OP}{\mathsf{OP}}
\newcommand{\Ctx}{\ensuremath{\mathsf{Ctx}}}
\newcommand{\DB}{\mathsf{DB}}

\newcommand{\mng}{\mathsf{mng}}
\newcommand{\br}{\mathsf{br}}
\newcommand{\KB}{\mathsf{KB}}
\newcommand{\kb}{\mathsf{kb}}

\newcommand{\head}{\mathsf{head}}
\newcommand{\pnot}{\mathsf{not}\ }
\newcommand{\U}{\mathcal U}
\newcommand{\V}{\mathcal V}

\newcommand{\altern}[1]{\ensuremath{\mathsf{vrt}(#1)}}
\newcommand{\m}[1]{\ensuremath{\mathsf{#1}}}
\newcommand{\aicarrow}{\Longrightarrow}

\begin{document}

\title{Active Integrity Constraints for Multi-Context Systems%
\thanks{This work was supported by the Danish Council for Independent
Research, Natural Sciences, grant DFF-1323-00247, and by FCT/MCTES/PIDDAC under centre grant
to BioISI (Centre Reference: UID/MULTI/04046/2013).
}}
\author{Lu\'\i s Cruz-Filipe\inst{1}
  \and Gra\c ca Gaspar\inst{2}
  \and Isabel Nunes\inst{2}
  \and Peter Schneider-Kamp\inst{1}}
\institute{
  Dept.\ of Mathematics and Computer Science,
  University of Southern Denmark
  \and BioISI---Biosystems \&\ Integrative Sciences Institute, Faculty of Sciences, University of Lisbon}

\maketitle

\begin{abstract}
  We introduce a formalism to couple integrity constraints over general-purpose
  knowledge bases with actions that can be executed to restore consistency.
  This formalism generalizes active integrity constraints over databases.
  In the more general setting of multi-context systems, adding repair suggestions to integrity constraints
  allows defining simple iterative algorithms to find all possible grounded repairs -- repairs for the
  global system that follow the suggestions given by the actions in the individual rules.
  We apply our methodology to ontologies, and show that it can express most
  relevant types of integrity constraints in this domain.
\end{abstract}

\section{Introduction}
\label{sec:intro}

Integrity constraints (ICs) for databases have been an important topic of research since the 1980s~\cite{Abiteboul1988}.
An early survey~\cite{Thalheim1991} already identified over $90$ relevant types of integrity constraints.
Since then, significant effort has been focused not only on identifying inconsistencies, but also on repairing inconsistent databases.

The same problem has been studied in other domains of knowledge representation. Integrity constraints for deductive databases~\cite{Asirelli1985} were also considered in the 1980s.
More recently, interest for integrity constraints has arisen in the ontology domain, with several approaches on how to define them and how to check their satisfaction~\cite{Fang2013,Motik2009,Ouyang2013}. Given its challenges, the more complex problem of repairing inconsistent knowledge bases has not received as much attention.

In this paper, we address the problem of computing repairs by combining two ideas: clausal-form integrity constraints for multi-context systems (MCSs)~\cite{CNS16} and active integrity constraints (AICs) for relational databases~\cite{Flesca2004}. We demonstrate the expressiveness of our formalism and show how it can be used to compute repairs for inconsistent MCSs in general, and for ontologies, in particular.

\paragraph{Contribution.}
The main contribution of this paper is a notion of AIC for MCSs, which enables us to compute repairs for inconsistent MCSs automatically, requiring only decidability of entailment in the individual contexts. Particularized to ontologies, our framework is expressive enough to capture all types of integrity constraints identified as relevant in \cite{Fang2013}, as we exemplify in the text.

The step from ICs for MCSs to AICs for MCSs is inspired by the similar step in the database case \cite{Flesca2004}. However, we draw more significant benefits in this more general setting. AICs are ICs that also specify possible repair actions in their head. In the database case, every clausal IC can be transformed into an AIC automatically. The goal, though, is to restrict in order to establish preferences among different possible repairs. In the general case, such a transformation would require solving complex abduction problems~\cite{Guessoum1998}.

Using AICs, we can automatically compute repairs for inconsistent MCSs, bypassing the need to solve such reasoning problems. The price to pay is the need to prove that an AIC is valid (Definition~\ref{defn:aic-valid}). The key observation here is that AICs should be written with a very clear semantic idea in mind, typically by an engineer with a deep knowledge of the underlying system, who should be able to show their validity formally. Thus, in practice, the complexity involved in computing each repair is moved to a one-time verification of validity of AICs.

\paragraph{Structure.}
We review previous work in Section~\ref{sec:backgr}, summarizing the key notions from~\cite{Caroprese2006,lcf16,CNS16}.
Section~\ref{sec:aic} introduces AICs for MCSs, showing that
they generalize the corresponding notion for relational databases, and studies their properties in general.
Section~\ref{sec:apps} focuses on the case of ontologies and evaluates our formalism against the classes of integrity constraints identified in \cite{Fang2013}.
Section~\ref{sec:algorithms} discusses how algorithms to compute repairs in the database setting can be
adapted to the general case of MCSs.
We conclude in Section~\ref{sec:concl}.

\subsection{Related work}
\label{sec:rw}

\paragraph{Database repairs.}
ICs for databases have been extensively studied throughout the last decades, and
we restrict ourselves to works most directly related to ours.

Integrity constraints are typically grouped in different syntactic categories~\cite{Thalheim1991}.
Many important classes can be expressed as first-order formulas, and can also be written in denial
(clausal) form -- the fragment expressable in our formalism.

Whenever an integrity constraint is violated, the database must be \emph{repaired} to regain consistency.
The problem of database repair is to determine whether such a transformation is possible, and many authors
have invested in algorithms for computing database repairs efficiently.
Typically, there are several possible ways of repairing an inconsistent database, and several criteria have been
proposed to evaluate them.
Minimality of change~\cite{Eiter1992,Winslett1990} demands that the database be changed as little as possible,
while the common-sense law of inertia~\cite{Przymusinski1997} states that every change should have an
underlying reason.
While these criteria narrow down the possible database repairs, it is commonly accepted that human interaction
is ultimately required to choose the ``best'' possible repair~\cite{Teniente1995}.

\paragraph{Active integrity constraints (AICs).} The formalism of AICs, introduced in~\cite{Flesca2004}, addresses
the issue of choosing among several possible repairs.
An AIC specifies not only an integrity constraint, but it also gives indications on how inconsistent databases
can be repaired through the inclusion of \emph{update actions}, which can be addition and removal of tuples
from the database -- a minimal set that can implement the three main operations of database updates~\cite{Abiteboul1988}.

The original, declarative, semantics of AICs defined \emph{founded} repairs~\cite{Caroprese2006}, in which every action is \emph{supported}: it occurs in the head of a constraint that is violated if
that action is not included.
Despite this characterization, there are unnatural founded repairs
where two actions mutually support each other, but do not have support from other actions.
The same authors then proposed \emph{justified} repairs~\cite{Caroprese2011}, which however are not intuitive and pose further problems~\cite{CEGN2013}.
Furthermore, justified repairs are intrinsically linked to the syntactic structure of
databases, and cannot be adapted to other knowledge representation formalisms.

\emph{Grounded repairs}~\cite{lcf16} form a middle ground between both semantics, requiring
support
for arbitrary subsets of the repair.
They are grounded fixed points of the intuitive operation of
``applying one action from the head of each AIC that is not satisfied'', which is in line with the intuitive
motivation for studying AICs.

Founded and justified repairs can be computed via revision
programming~\cite{Caroprese2011}.
Alternatively, an operational semantics for AICs~\cite{CEGN2013} was implemented for SQL
databases~\cite{KMIS2015}.
There, repairs are leaves of particular trees, yielding a semantics equivalent to the declarative one when existence of a repair is an NP-complete
problem.
For grounded and justified repairs, where this existence problem is $\Sigma^P_2$-complete, the
trees still contain all repairs, but may also include spurious leaves -- requiring a post
test that brings the overall complexity to the theoretical limit.

\paragraph{Multi-context systems (MCSs).}
MCSs, as defined in~\cite{Brewka2007b}, can be informally described as
collections of logic knowledge bases -- the \emph{contexts} -- connected by Datalog-style
\emph{bridge rules}.
Since their introduction, several variants of MCSs have been proposed that add to their potential fields of
application.
Relational MCSs~\cite{Fink2011} were proposed as a way to allow a formal first-order syntax, introducing
variables and aggregate expressions in bridge rules, and extending the semantics of MCSs accordingly.
Managed MCSs~\cite{Brewka2011}, which we describe in Section~\ref{sec:MCSs}, further generalize MCSs by
abstracting from the possible actions that change individual knowledge bases. Other variants, which are not directly relevant for this work, are discussed in~\cite{CNS16}.
%
%
A different line of research deals with repairing \emph{logical} inconsistency of an MCS
(non-existence of a model)~\cite{Eiter2011b}.

\paragraph{ICs in ontologies.}
Integrating ICs with ontology-based systems poses several challenges, mainly due to the open-world assumption and the absence of the unique name assumption~\cite{Fang2013,Motik2010,Patel-Schneider2012,Tao2010}.
In this context, ICs are conventionally modelled as T-Box axioms~\cite{Motik2009}, but variants based on hybrid knowledge bases, auto-epistemic logic, modal logic, and grounded circumscription have recently been proposed. For an overview of these proposals see Section~2 in~\cite{Ouyang2013}. For details on how some of these can be expressed by ICs over MCSs, using a systematic interpretation of ontologies as MCSs, see Section~4.5 in~\cite{CNS16}. The interpretation we use in Section~\ref{sec:apps} is a variant of the one presented therein.

\section{Background}
\label{sec:backgr}


\paragraph{AICs for databases.}
\label{sec:ICsandAICs}

Let $\Sigma$ be a first-order signature without function symbols.
A database is a set of ground atoms over $\Sigma$, and an update action is an expression of the form $+a$ or
$-a$, where $a$ is a ground atom over $\Sigma$.
An \emph{active integrity constraint} (AIC) over a database $\DB$ is a rule $r$ of the form
\begin{equation}
  \label{eq:aic-db}
  p_1,\ldots,p_m,\pnot(p_{m+1}),\ldots,\pnot(p_\ell)\\
  \aicarrow
  \alpha_1\mid\cdots\mid\alpha_k
\end{equation}
where each $p_i$ is an atom over the database's signature, every variable 
free in $p_{m+1},\ldots,p_\ell$ occurs in $p_1,\ldots,p_m$, and each
update action $\alpha_i$ is either $-p_j$ for some $1\leq j\leq m$ or $+p_j$ for
$m<j\leq\ell$.\footnote{In~\cite{Flesca2004}, existentially quantified
  variables can also occur in negative literals. This was not discussed in subsequent work, and we
  ignore it for simplicity of presentation.}
The \emph{body} of $r$ is $\body(r)=p_1,\ldots,p_m,\pnot(p_{m+1}),\ldots,\pnot(p_\ell)$, and the
\emph{head} of $r$ is $\head(r)=\alpha_1\mid\ldots\mid\alpha_k$.

If $r$ is ground, then $\DB$ \emph{satisfies} $r$, denoted $\DB\models r$, if $\DB\not\models p_i$ for some
$1\leq i\leq m$ or $\DB\models p_i$ with $m<i\leq\ell$.
In general, $\DB\models r$ if $\DB$ satisfies all ground instances of $r$.
Otherwise, $r$ is \emph{applicable} in $\DB$~\cite{Flesca2004}.
If $\eta$ is a set of AICs, then $\DB\models\eta$ if $\DB\models r$
for every $r\in\eta$.

A set of update actions $\U$ is \emph{consistent} if it does not
contain both $+a$ and $-a$ for any ground atom $a$.
Given a consistent $\U$, we write $\U(\DB)$ for the result of applying all actions in $\U$ to
$\DB$, and say that $\U$ is a
\emph{weak repair} for $\langle\DB,\eta\rangle$ if:
(i)~every action in $\U$ changes $\DB$ and (ii)~$\U(\DB)\models\eta$.
$\U$ is a \emph{repair} if $\V(\DB)\not\models\eta$ for every
$\V\subsetneq\U$~\cite{Caroprese2006}, and
$\U$ is \emph{grounded} if, for every $\V\subsetneq\U$, there exists a ground instance $r$ of a
rule in $\eta$ such that $\V(\DB)\not\models r$ and $\head(r)\cap(\U\setminus\V)\neq\emptyset$~\cite{lcf16}.


\paragraph{Multi-Context Systems.}
\label{sec:MCSs}
We now describe the variant of multi-context systems we use: \emph{managed multi-context systems} (also abbreviated to MCSs)~\cite{Brewka2011}.

A \emph{relational logic} $L$ is a tuple $\langle\KB,\BS,\ACC,\Sigma\rangle$, where $\KB$ is the
set of well-formed knowledge bases of $L$ (sets of well-formed formulas), $\BS$ is a set of possible belief
sets (candidate models), $\ACC:\KB\to 2^{\BS}$ is a function assigning to each knowledge base a set of acceptable
belief sets (its models), and $\Sigma$ is a signature 
generating first-order sublanguages of $\bigcup\KB$ and $\bigcup\BS$.

A \emph{managed multi-context system} is a collection of managed contexts $\{C_i\}_{i=1}^n$, with each
$C_i=\langle L_i,\kb_i,\br_i,D_i,\OP_i,\mng_i\rangle$ where: $L_i=\langle\KB_i,\BS_i,\ACC_i,\Sigma_i\rangle$
is a relational logic; $\kb_i\in\KB_i$; $D_i$ (the \emph{import domain}) is a set of constants from
$\Sigma_i$; $\OP_i$ is a set of operation names; $\mng_i:\wp(\OP_i\times\bigcup\KB_i)\times\KB_i\to\KB_i$
is a \emph{management function}; and $\br_i$ is a set of \emph{managed bridge rules}, with the form
\begin{equation}
  \label{eq:bridge}
  (i:o(p)) \leftarrow (i_1:p_1),\ldots,(i_q:p_q), \pnot (i_{q+1}:p_{q+1}),\ldots,\pnot(i_m:p_m)
\end{equation}
such that $o\in\OP_i$, $p\in\bigcup\KB_i$, $1\leq i,i_j\leq n$, and each $p_j$ is a
belief\footnote{Technically, $P_p$ is a \emph{relational element} of $C_{i_p}$: it can include variables,
  which when instantiated yield elements of $\bigcup\BS_{i_p}$ -- see~\cite{Brewka2011} for details.} of
$L_{c_j}$.

Intuitively, $\kb_i$ is the knowledge base of context $C_i$ and $\OP_i$ are the names of the operations that
can be applied to change it.
The management function defines the semantics of these operations: $\mng_i(O,\kb)$ is the result of applying
the operations in $O$ to $\kb$.
Bridge rules govern the interaction between contexts.%
\footnote{For the sake of presentation, we simplified the management function, which in the original work is allowed to return several possible effects for each action.}

A \emph{belief state} for an MCS $M=\{C_i\}_{i=1}^n$ is a set $S=\{S_i\}_{i=1}^n$ such that
each $S_i\in\BS_i$.
A ground instance of bridge rule~\eqref{eq:bridge} is \emph{applicable} in $S$ if $p_i\in S_i$ for
$1\leq i\leq q$ and $p_i\not\in S_i$ for $q<i\leq m$; the variables in the rule can only be instantiated by
elements of the import domain $D_i$.
A belief state is an \emph{equilibrium} for $M$ if it is stable under application of all bridge rules, i.e.:
\[S_i\in\ACC_i(\mng_i(\{\head(r)\mid r \in \br_i \mbox{ applicable in }S\},\kb_i))\]
In general, $M$ can have zero, one or several equilibria; if at least one exists, then $M$
is \emph{logically consistent}.
We present examples of MCSs in the next sections.

\paragraph{Integrity constraints for general-purpose knowledge bases.}
\label{sec:generalKBs}  

ICs for MCSs~\cite{CNS16} generalize clausal ICs to a
generic framework for reasoning systems -- covering not only relational databases, but also deductive
databases, peer-to-peer systems and ontologies, among others.
Syntactically, ICs are bridge rules with empty head, forming an added layer on top of an MCS that does not affect its semantics.

As MCSs may have several equilibria,
satisfaction of a set of ICs $\eta$ can be \emph{weak} --
there is an equilibrium satisfying all rules in $\eta$ -- or \emph{strong} --
all equilibria satisfy all rules in $\eta$.
In order to avoid vacuous quantifications, strong satisfaction only holds for logically consistent MCSs.
In general these properties are undecidable~\cite{CNS16},
but if entailment in every context is decidable then satisfaction of a set of ICs is in most cases as hard as
the hardest entailment decision problem.

In this paper, we do not explicitly mention the set of
ICs when clear from the context.
Moreover, our development applies both to weak and strong satisfaction, and
we simply say that an MCS is \emph{consistent} if it satisfies
the given set of ICs.
We explicitly write ``logical consistency'' for existence of an equilibrium.

\section{Active Integrity Constraints}
\label{sec:aic}

We begin by defining active integrity constraints over multi-context systems.

\begin{definition}
  An AIC over an MCS $M=\{C_i\}_{i=1}^n$ is a rule $r$ of the form
  \begin{multline}
    \label{eqn:aic}
    (i_1:P_1),\ldots,(i_m:P_m),\pnot(i_{m+1}:P_{m+1}),\ldots,\pnot(i_\ell:P_\ell)\\
    \aicarrow
    (j_1:\alpha_1)\mid\cdots\mid(j_k:\alpha_k)
  \end{multline}
  where $1\leq i_p,j_q\leq n$, each $P_p$ is a
  belief in $C_{i_p}$, each \emph{update action} $\alpha_q\in\OP_{j_q}\times\bigcup\KB_{j_q}$, and
  all variables in $P_{m+1},\ldots,P_\ell$ occur in $P_1,\ldots,P_m$.
\end{definition}

This definition follows the one for databases~\eqref{eq:aic-db}, and we define \emph{body} 
and \emph{head} of $r$ similarly.
Equation~\eqref{eqn:aic} also generalizes ICs for MCSs: each AIC corresponds to an IC by
ignoring its head, immediately yielding notions of weak and strong satisfaction for an AIC.
We also say that $r$ is \emph{applicable} to an MCS $M$ if $M\not\models r$.
Intuitively, in this case $M$ should be repaired by applying actions in $\head(r)$.

The reasoning capabilities of MCSs dictate that we cannot restrict the actions in the head of an AIC syntactically (as in the database world,
see Section~\ref{sec:ICsandAICs}).
We thus relax this requirement by only demanding that the actions are capable of solving the inconsistency.
It is also not reasonable to require that every action in $\head(r)$ be able to solve every inconsistency
detected by $\body(r)$: since inconsistencies may be triggered by derived information, they may have
different origins, and the different actions may be solutions for those different causes.

We are interested in sets of update actions that are applied simultaneously, i.e.~the order in which actions are executed should be irrelevant.
This corresponds to the consistency requirement usually considered in databases.%
\begin{definition}
  Let $M=\{C_i\}_{i=1}^n$ be an MCS, $\U$ be a finite set of update actions, and
  $\U_i$ be the set of actions in $\U$ affecting $C_i$.

  $\U_i$ is \emph{consistent} w.r.t.~$\kb_i$ if, for every permutation $\alpha_1,\ldots,\alpha_k$ of the elements of $\U_i$,
  $\mng_i(\U_i,\kb_i)=\mng_i(\alpha_1,\mng_i(\ldots,\mng_i(\alpha_k,\kb_i)\ldots))$.
  $\U$ is consistent w.r.t.~$M$ if each $\U_i$ is consistent w.r.t.~$\kb_i$, and in this
  case we write $\U(M)$ for the result of applying each $\U_i$ to each $\kb_i$.
\end{definition}

\begin{example}
  \label{ex:toy}
  We consider a concrete toy example of a deductive database with two unary base relations $p$ and $q$, a view
  consisting of a relation $r$ such that $r(x)\leftrightarrow p(x)\vee q(x)$, and the integrity constraint
  $\neg r(a)$.

  We formalize this as an MCS $M=\langle C_E,C_I\rangle$ where $C_E$ is an extensional database including
  predicates $p$ and $q$ (but not $r$), $C_I$ is the view context including predicate $r$ (but not $p$ or
  $q$), and they are connected by the bridge rules
  \[(I:r(X)) \leftarrow (E:p(X)) \qquad (I:r(X)) \leftarrow (E:q(X))\,.\]
  Furthermore, $\mng_E$ allows addition and removal of any tuples to $C_E$, using operations $\m{add}$
  and $\m{del}$, while $\mng_I$ does not allow any changes.
  (See~\cite{CNS16} for details of this construction.)

  From the structure of $M$, we know that $r(a)$ can only arise as a deduction from $p(a)$ or $q(a)$
  (or both), so it makes sense to write an AIC
  \[(I:r(a)) \aicarrow (E:\mathsf{del}(p(a)))\mid(E:\mathsf{del}(q(a)))\,.\]
  The actions on the head of this AIC solve the problem in all future states of $M$,
  since $C_I$ cannot change.
  However, restoring consistency may require performing both actions (if the database contains both $p(a)$ and $q(a)$).
\end{example}

This example also illustrates an important point: repair actions are written with a particular structure of the MCS in
mind.

\begin{definition}
  The set of \emph{variants} to an MCS $M$, denoted $\altern M$, is
  \[\altern M=\{\U(M)\mid\U\mbox{ is a finite set of update actions over $M$}\}\,.\]
\end{definition}

Restrictions on the actions in the head of AICs only range over $\altern M$, which contains all possible
future evolutions of $M$.

\begin{definition}
  \label{defn:aic-valid}
  An AIC $r$ of the form~\eqref{eqn:aic} is \emph{valid} w.r.t.~an MCS $M$ if:
  \begin{itemize}
  \item for every logically consistent $M'\in\altern M$ such that $M'\not\models r$,
    there is $\U\subseteq\head(r)$ with $\U(M')\models r$;
  \item for every $\alpha\in\head(r)$, there is $M'\in\altern M$ with $M'\not\models r$ and
    $\alpha(M')\models r$.
  \end{itemize}
\end{definition}
These conditions require that the set of suggested actions be complete (it can solve all inconsistencies) and
that it does not contain useless actions.

\begin{example}
The AIC in Example~\ref{ex:toy} is valid: the only possible changes to $M$ are in $\kb_E$, which only contains
information about $p$ and $q$, thus, in any element of $\altern M$ the only way to derive $r(a)$ is still from
either $p(a)$ or $q(a)$.
The second condition follows by considering $M'$ with $\kb_E=\{p(a)\}$ and $\kb_E=\{q(a)\}$.
\end{example}

\begin{proposition}
  Deciding whether an AIC is valid is in general undecidable.
\end{proposition}
\begin{proof}[sketch]
  Let $L$ be a logic with an undecidable entailment problem, $C$ be a context over $L$ with $\m{add}\in\OP_C$
  such that $\mng_C(\m{add}(\varphi),\Gamma)=\Gamma\cup\{\varphi\}$, and $M=\{C\}$.
  Assume also that $\altern M$ includes all knowledge bases over $L$.
  Then $(C:\neg B)\aicarrow(C:\m{add}(A))$ is valid iff $A\models_LB$.
  \qed
\end{proof}
 
In practice, proving validity of AICs should not pose a problem: AICs are written by humans with a very
precise semantic motivation in mind, and this means that the conditions in Definition~\ref{defn:aic-valid}
should be simple for a human to prove.

We now show that the framework we propose generalizes the database case.
A database $\DB$ can be seen as an MCS $M(\DB)$, defined as having a single context over first-order logic, whose knowledge base is $\DB$, with management function allowing addition ($+$) or removal ($-$) of facts, and where
the only set of beliefs admissible w.r.t.~a given database is the set of literals that are true in that database
(see~\cite{CNS16} for a detailed definition).
\begin{proposition}
  Every AIC over a database $\DB$ yields a valid AIC over $M(\DB)$.
\end{proposition}
\begin{proof}[sketch]
  We write a generic AIC over a database~\eqref{eq:aic-db} as the AIC
  \[
    (1:p_1),\ldots,(1:p_m),\pnot(1:p_{m+1}),\ldots,\pnot(1:p_\ell)
    \aicarrow
    (1:\alpha_1)\mid\cdots\mid(1:\alpha_k)
  \]
  over $M(\DB)$.
  If $\DB$ does not satisfy the body of~\eqref{eq:aic-db}, then it can always be repaired by
  performing exactly one of the actions in its head~\cite{Caroprese2008}, establishing both conditions for
  validity.\qed
\end{proof}

\begin{definition}
  Let $M=\{C_i\}_{i=1}^n$ be an MCS, $\eta$ be a set of AICs over $M$ and $\U$ be a finite set of update actions.
  $\U$ is a \emph{weak repair} for $\langle M,\eta\rangle$ if $\U$ is consistent w.r.t.~$M$ and
  $\U(M)\models\eta$.
  Furthermore, $\U$ is \emph{grounded} if: for every $\V\subsetneq\U$, there is an
  AIC $r\in\eta$ such that $\V(M)\not\models r$ and $\head(r)\cap(\U\setminus\V)\neq\emptyset$.
\end{definition}
The definitions of weak and grounded repair directly correspond to those for the database case
(Section~\ref{sec:ICsandAICs}).
The notion of grounded repair implies, in particular, minimality under inclusion~\cite{lcf16}.

\section{Application: the Case of Ontologies}
\label{sec:apps}

This section is devoted to examples illustrating how our framework can be applied to the particular case of
integrity constraints over ontologies.

Previous work~\cite{Brewka2007b,CNS16} shows how to view an ontology as a context of an MCS.
In the present work, we refine this interpretation by representing an ontology as \emph{two} contexts:
one for the A-Box, one for the T-Box, connected by bridge rules that port every
instance from the former into the latter.
(This is reminescent of how deductive databases are encoded in MCSs, see~\cite{CNS16}.)
This finer encoding allows us, in particular, to reason about asserted instances (which are given in the
A-Box) and those that are derived using the axioms (see Example~\ref{ex:indiv}).

We further assume that the A-Box only contains instances of atomic concepts or roles ($C(t)$ or $R(t,t')$).
This option does not restrict the expressive power of the ontology, but it helps structure AICs: to include
instance axioms about e.g.~$C\sqcup D$, one instead defines a new concept $E=C\sqcup D$ in the T-Box and
includes instance axioms about $E$ in the A-Box (see also Example~\ref{ex:explicit}).

\begin{definition}
  A description logic $\mathcal L$ is represented as the relational logic
  $L_{\mathcal L}=\langle\KB_{\mathcal L},\BS_{\mathcal L},\ACC_{\mathcal L},\Sigma_L\rangle$, where:
  \begin{itemize}
  \item $\KB_{\mathcal L}$ contains all well-formed knowledge bases of~$\mathcal L$;
  \item $\BS_{\mathcal L}$ contains all sets of queries in the language of $\mathcal L$;
  \item $\ACC_{\mathcal L}(\kb)$ is the singleton set containing the set of queries to which $\kb$ answers
    ``Yes''.
  \item $\Sigma_{\mathcal L}$ is the first-order signature underlying $\mathcal L$.
  \end{itemize}
  An ontology $\mathcal O=\langle T,A\rangle$ based on $\mathcal L$ induces the multi-context system $M(\mathcal O)=\langle\Ctx(T),\Ctx(A)\rangle$
  where $\Ctx(T)=\langle L_{\mathcal L},T,\br_T,\Sigma_0,\emptyset,\emptyset\rangle$ with
  \begin{itemize}
  \item $\br_T$ contains all rules of the form $(T:C)(X) \leftarrow (A:C)(X)$ where $C$ is a concept, and
    $(T:R)(X,Y) \leftarrow (A:R)(X,Y)$ where $R$ is a role;
  \item $\Sigma_0$ is the set of constants in $\Sigma_{\mathcal L}$;
  \end{itemize}
  and $\Ctx(A)=\langle L_{\mathcal L},A,\emptyset,\Sigma_0,OP,\mng\rangle$ where $OP$ and $\mng$ are the set
  of allowed update operation names and their definition.
\end{definition}
The management function does not allow changes to the T-Box; the particular operations in the A-Box depend
on the concrete ontology.
This is in line with our motivation that writing AICs requires knowledge of the system's deductive abilities
(expressed by the T-Box), which should not change.

We now evaluate the expressivity of our development by showing how to formalize several types of ICs over
ontologies.
We follow the classification in Section~4.5 of~\cite{Fang2013}, which describes families of ICs determined by
OWL engineers and ontologists as the most interesting, as well as other types of ICs considered in the
scientific literature.
Several classes of ICs are syntactically similar, so we do not include examples for all categories
in~\cite{Fang2013}, but explain in the text how the missing ones can be treated.

Most of our examples are adapted from~\cite{Fang2013}, which frames them in a variant of the Lehigh University
Benchmark~\cite{Guo2005}, an ontology designed with the goal of providing a realistic scenario for testing.
This ontology considers concepts \m{student}, \m{gradStudent}, \m{class} and \m{email}, and roles
\m{hasEmail}, \m{enrolled} and \m{webEnrolled}.
Our semantics is: \m{class} is a concept including all classes of a common course; \m{enrolled(c,s)} holds if
student \m{s} is enrolled in course \m{s}; and \m{webEnrolled} holds if the student is furthermore to be
contacted only electronically.\footnote{This semantics is slightly changed from that of~\cite{Fang2013}, in
  order to make some aspects of our example more realistic.}
The actual contents of the A-Box are immaterial for our presentation, and we restrict ourselves to the fragment of the T-Box containing the following axioms.
\begin{align*}
  \m{gradStudent}&\sqsubseteq\m{student}
  & \exists\m{enrolled}.\m{student}&\sqsubseteq\m{class} \\
  \m{webEnrolled}&\sqsubseteq\m{enrolled}
  & \exists\m{hasEmail}.\m{email}&\sqsubseteq\m{student} \\
  && \exists\m{webEnrolled}^R.\m{class}&\sqsubseteq\exists\m{hasEmail}
\end{align*}

\subsection{Functional dependencies}
Functional dependencies are one of the most frequently occurring families of ICs: requirements that certain
relations be functional on one argument.
In our example, this applies to \m{hasEmail}: two distinct students cannot have the same e-mail.

Since ontologies do not have the Unique Name Assumption, we cannot distinguish individuals by checking name
equality (as in databases), but must query the ontology instead.
Furthermore, while in the database world such violations can only be repaired by removing one of the offending
instances, in ontologies, we can also add the information that two individuals are the same.

\begin{example}
  \label{ex:functional}
  Suppose that the management function includes operations \m{add} and \m{del} to add or remove a particular
  instance from the A-Box, as well as \m{assertEqual}, establishing equality of two individuals.
  Under these assumptions, we can express funcionality of e-mail as the following AIC.
  \begin{multline}
    \label{aic:functional}
    (A:\m{hasEmail}(X,Z)),(A:\m{hasEmail}(Y,Z)),\pnot(T:(X=Y))\\
    \aicarrow
    (A:\m{del}(\m{hasEmail}(X,Z)))\mid(A:\m{assert}(X=Y))
  \end{multline}
  Observe that, if $T$ explicitly proves that $X\neq Y$, then only the first action can be used, as asserting
  equality between $X$ and $Y$ would lead to an inconsistency.
  However, if this is not the case then the second action is also a repair possibility, and hence this AIC is
  valid.
  There are several possibilities for the implementation of \m{assert}: it can add the equality $X=Y$ to the
  A-Box, but it can also syntactically replace every occurrence of one of them for the other.
\end{example}

Several other types of dependencies (e.g.~key constraints, uniqueness constraints, functionality constraints)
are expressed by similar formulas.
Likewise, max-cardinality constraints can be represented as AICs with similar types of actions in the head
(deleting some instances or unifying some individuals).

\subsection{Property domain constraints}

This family of ICs specifies that the domain of a role should be a subset of a particular concept.
In case such a constraint is violated, the offending element has to be added as an instance of that concept.
The treatment of these ICs is thus very similar to the database case.

\begin{example}
  \label{ex:property}
  To model that only students can be enrolled in courses, we write the following AIC.
  \begin{equation}
  \label{aic:property}
  (T: \m{enrolled}(X,Y), \pnot(T: \m{student}(Y))
  \aicarrow
  (A: \m{add}(\m{student}(Y)))
  \end{equation}
\end{example}
We could also add the action $(A:\m{del}(\m{enrolled}(X,Y)))$ to the head of this AIC; note that it would only
restore consistency in the case where this fact is explicitly stated in the A-Box and not otherwise derivable.
Property range constraints (restricting the range of a role) can be similarly treated.

\subsection{Specific type constraints}
In many applications, it is interesting to minimize redundancy in the A-Box.
In particular, in the presence of inclusion axioms, it is often desirable only to include instances pertaining
to the most specific type class of each individual.

\begin{example}
  \label{ex:indiv}
  Since $\m{gradStudent}\sqsubseteq\m{student}$, we guarantee that the A-Box only contains
  instances of the most specific class a student belongs to by writing:
  \begin{equation}
    \label{aic:indiv}
    (A:\m{gradStudent}(X)),(A:\m{student}(X))
    \aicarrow
    (A:\m{del}(\m{student}(X)))
  \end{equation}
\end{example}

Thus, if the A-Box contains e.g.~\m{student(john)} and \m{gradStudent(john)}, then the axiom \m{student(john)}
will be removed.
The system will still be able to derive \m{student(john)}, but only in context $C_T$ (using the information in
the T-Box).
The separation of the A-Box and T-Box in different contexts is essential to express this
integrity constraint in our formalism.
Constraints that distinguish between assertions explicitly stated in the A-Box and derived ones have been
considered e.g.~in~\cite{Patel-Schneider2012}.

\subsection{Min-cardinality constraints}
We now consider a more interesting type of ICs: min-cardinality constraints.
Inconsistencies arising from the violation of such constraints are hard to repare automatically, as such a
repair requires ``guessing'' which instances to add.
Using AICs and adequate management functions, we can even specify the construction of ``default'' values that
may depend on the actual ontology.

\begin{example}
  \label{ex:min-card}
  We want to express that each class must have a minimum of $10$ students.
  Classes with less enrolled students should be closed, and those students moved to the smallest remaining
  class using an operation \m{redistribute}.
  \begin{equation}
    \label{aic:min-card}
    (T:(\leq 10.\m{enrolled})(X))
    \aicarrow
    (A:\m{redistribute}(\neg\m{class}(X)))
  \end{equation}
\end{example}
For this AIC to be valid, \m{redistribute} must check whether students are \m{enrolled} or
\m{webEnrolled} and change the appropriate instance in the A-Box.
This also uses the knowledge that instances of \m{enrolled} cannot be derived in other ways.

A similar kind of constraints are totality constraints, which require that a role be total on one of its
arguments.
In our example, we could require every student to be enrolled in some class, and use an adequate management
function to add non-enrolled students to e.g.~the smallest class.

\subsection{Missing property value constraints}
We now turn our attention to a kind of ICs that is also very
common in ontologies: disallowing unnamed individuals for particular properties~\cite{Patel-Schneider2012}.

\begin{example}
  \label{ex:explicit}
  Our ontology specifies that all students that are \m{webEnrolled} in a class must have an e-mail
  address.
  However, for the purpose of contacting these individuals, this e-mail address must be explicitly provided.
  We address this issue with the following AIC.
  \begin{multline}
    \label{aic:explicit}
    (T:(\exists\m{hasEmail})(X)),\pnot(T:\m{hasEmail}(X,Y))\\
    \aicarrow
    (A:\m{unregister}(\neg\exists\m{webEnrolled}^R(X)))
  \end{multline}
  Here, \m{unregister} replaces the axiom $\m{webEnrolled}(X)$ with $\m{enrolled}(X)$, as it makes sense
  to keep the student enrolled in the course.
  Validity of this AIC follows from observing that the only possible ways to derive $\exists\m{hasEmail}(X)$
  are either from an explicit assertion $\m{hasEmail}(X,Y)$ or indirectly from $\m{webEnrolled}(Z,X)$.
\end{example}

This example also justifies our requirement that the A-Box can only contain instances of atomic concepts or
roles.
If the A-Box were allowed to contain e.g.~$\exists\m{hasEmail(john)}$, then AIC~\eqref{aic:explicit} would
no longer be valid.
By restricting to atomic concepts, the only way to perform a similar change would be by defining a new concept
as equivalent to $\exists\m{hasEmail}$ -- and this information would be present in the T-Box, making it clear
that AICs should consider it.

\subsection{Managing unnamed individuals}
Finally, we illustrate how we can write AICs in different ways to control whether they range over all
individuals of a certain class, or only over named ones.

\begin{example}
  \label{ex:unnamed}
  For ecological reasons, we want all students with an e-mail address to be enrolled in the web version of
  courses.
  We can write this as follows.
  \begin{multline}
    \label{aic:named}
    (T:(\m{hasEmail})(Y,Z)),(T:\m{enr}\lefteqn{\m{olled}(X,Y)),\pnot(T:\m{webEnrolled}(X,Y))}\\
    \aicarrow(A:\m{webEnroll}(\m{webEnrolled}(X,Y)))
  \end{multline}
  Operation \m{webEnroll} will replace $\m{enrolled}(X,Y)$ with $\m{webEnrolled}(X,Y)$, dually to
  \m{unregister} in the previous example.

  Alternatively, we could consider writing
  \begin{multline}
    \label{aic:unnamed}
    (T:(\exists\m{hasEmail})(Y)),(T:\m{enr}\lefteqn{\m{olled}(X,Y)),\pnot(T:\m{webEnrolled}(X,Y))}\\
    \aicarrow(A:\m{webEnroll}(\m{webEnrolled}(X,Y)))
  \end{multline}
  In this particular context, this formulation is undesirable, as it will also affect individuals who do not have a
  known e-mail address.
  By writing an explicit variable in the first query of the body, as in~\eqref{aic:named}, we guarantee that
  we only affect those individuals whose e-mail address is known.
\end{example}
Similar considerations about the two possible ways to formulate this type of ICs can be found
in~\cite{Patel-Schneider2012}.

\section{Computing Repairs}
\label{sec:algorithms}

In~\cite{CEGN2013}, we showed how to use active integrity constraints to compute repairs for inconsistent
databases, by using the actions in the head of unsatisfied AICs to build a \emph{repair tree} whose leaves
were the repairs.
We showed how the construction of the tree could be adapted to the different types of repairs considered
originally in~\cite{Caroprese2011}; in particular, for the case of grounded repairs (which is the one we are
interested in this work), it is enough to expand each node with the actions in the heads of the AICs that are
not satisfied in that node.

We adapt this construction to the framework of AICs over MCSs.
As we will see, the algorithms have to be adapted to this more general scenario, but we can still construct
all grounded repairs for a given (inconsistent) MCS automatically, as long as entailment in all contexts is
decidable.

\begin{definition}
  Let $M$ be an MCS and $\eta$ be a set of integrity contraints over $M$.
  The \emph{repair tree} for $\langle M,\eta\rangle$, $\mathcal T_{\langle M,\eta\rangle}$, is defined as follows.
  \begin{itemize}
  \item Each node is a set of update actions.
  \item A node $n$ is \emph{consistent} if: (i)~$n(M)$ is logically consistent and (ii)~if $n'$ is the parent
    of $n$, then $n$ is a consistent set of update actions w.r.t.~$n'(M)$.
  \item Each edge is labeled with a closed instance of a rule.
  \item The root of the tree is the empty set $\emptyset$.
  \item For each consistent node $n$ and rule $r$, if $n(M)\not\models r$ then
    $n'=n\cup\U$ is a child of $n$ if (i)~$\U\subseteq\head(r)$, (ii)~$n'(M)\models r$ and (iii)~if
    $\U'\subseteq\U$ then $(n\cup\U')(M)\not\models r$.
  \end{itemize}
\end{definition}

In the database case~\cite{CEGN2013}, it is straightforward to show that repair trees are finite, since the
syntactic restrictions on database AICs guarantee that each rule can only be applied at most once in every
branch.
In the general MCS case, this is not true, as the following example shows.

\begin{example}
  \label{ex:weird-repair}
  Consider an ontology (represented as an MCS as in Section~\ref{sec:aic}) with four concepts $B_1$, $B_2$,
  $B_3$ and $D$.
  The T-Box contains axioms
  \[B_1\sqsubseteq D \qquad \mbox{ and }\qquad B_2\sqcap B_3\sqsubseteq D\]
  and the A-Box is $\{B_1(a),B_3(a)\}$.
  Furthermore, we have integrity constraints
  \begin{align*}
    (T:D)(a) &\aicarrow (A:\m{del}(B_1)(a))\mid(A:\m{del}(B_3)(a)) & (r_1) \\
    \pnot(T:B_1)(a),\pnot(T:B_2)(a) &\aicarrow (A:\m{add}(B_2)(a)) & (r_2)
  \end{align*}
\end{example}
\begin{wrapfigure}[9]r{.5\textwidth}
  \vspace*{-3em}\[\xymatrix@R-1em{
    \emptyset\ar[d]^{r_1} \\
    \{\m{del}(B_1)(a)\}\ar[d]^{r_2} \\
    \{\m{del}(B_1)(a),\m{add}(B_2)(a)\}\ar[d]^{r_1} \\
    \{\m{del}(B_1)(a),\m{add}(B_2)(a),\m{del}(B_3)(a)\}
  }\]
\end{wrapfigure}

Following this construction, we obtain the tree on the right, and its leaf is a grounded repair.


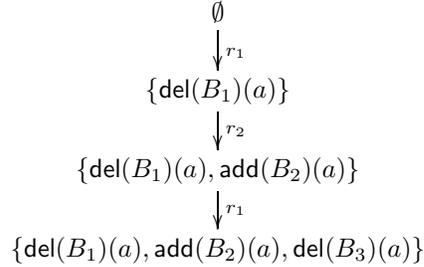
\begin{lemma}
  \label{lem:finite}
  $\mathcal T_{\langle M,\eta\rangle}$ is finite.
\end{lemma}
\begin{proof}
  By definition, every node of $\mathcal T_{\langle M,\eta\rangle}$ has a finite number of descendants, since
  there are only finitely many ground instances of AICs with a finite number of actions in each one's head.
  By construction, in every branch the labels of the nodes form an increasing sequence (w.r.t.~set
  inclusion), and each node is again a subset of the (finite) set of all actions in the heads of all rules.
  Therefore, $\mathcal T_{\langle M,\eta\rangle}$ has finite depth and finite degree, hence it is finite.\qed
\end{proof}

\begin{lemma}
  \label{lem:complete}
  Every grounded repair for $\langle M,\eta\rangle$ is a leaf of $\mathcal T_{\langle M,\eta\rangle}$.
\end{lemma}
\begin{proof}
  Let $\U$ be a grounded repair for $M$ and $\eta$.
  By definition of grounded repair, if $\U'\subseteq\U$ then there is a ground instance $r$ of an AIC such
  that: there exists $\V\subseteq\head(r)\cap\U$ such that $(\U'\cup\V)(M)\models r$.
  This directly yields a branch of the repair tree ending at $\U$.\qed
\end{proof}
(This is essentially the same argument for showing that, in the database case, grounded repairs are
well-founded, see~\cite{lcf16}.)

$\mathcal T_{\langle M,\eta\rangle}$ is constructed as the well-founded repair tree in the database
case~\cite{CEGN2013}.\footnote{There is also a notion of repair tree for databases in~\cite{CEGN2013}, but it
  relies on the ability of inferring heads of AICs automatically, which does not exist in the MCS setting.}
In both cases, this tree may, in general, contain leaves that are not grounded repairs~\cite{lcf16}.
Under the assumption that $\mbox{P}\neq\mbox{NP}$, this cannot be avoided, since existence of grounded repairs
for databases is a $\Sigma^P_2$-complete problem~\cite{lcf16}.

\paragraph{Complexity.}
The proof of Lemma~\ref{lem:finite} shows that the depth of $\mathcal T_{\langle M,\eta\rangle}$ is polynomial
in the size of the grounded instances of $\eta$.
Therefore, given an oracle that decides whether an MCS satisfies a set of AICs, the problem of existence of a
grounded repair for $\langle M,\eta\rangle$ is $\Sigma^P_2$-complete: $\mathcal T_{\langle M,\eta\rangle}$ can
be built in non-deterministic polynomial time (guessing which rule to apply at each node and using the oracle
to decide whether the descendant is a leaf), and the validation step can be done in co-NP time (if the leaf is
not a grounded repair, then we guess the subset that violates the definition and use the oracle to confirm
this).

\section{Discussion and Conclusions}
\label{sec:concl}



\paragraph{Validity.}
At the end of Example~\ref{ex:toy}, we pointed out that restoring consistent w.r.t.~an AIC $r$ may require
applying several actions in $\head(r)$.
This suggests allowing sets of actions (rather than actions) in the heads of AICs.
Besides increasing the complexity of our development, it is not clear that this change would bring significant
benefits.
In terms of computing repairs, we already cover those cases, since we add sets of actions when going from a
node to its descendents.
Also, it is not clear that there exists a situation when \emph{every} possible inconsistent MCS requires a set
of actions to repair.

One could also remove the second condition of validity of an AIC, i.e.~allow the actions in the head to be
insufficient to restore consistency of some MCSs.
This would remove some burden from the programmer who has to specify the AICs, and would not affect the
performance of the algorithms in Section~\ref{sec:algorithms}.
However, it would contradict the original motivation for AICs~\cite{Flesca2004}: that the actions in the head
of a rule should provide the means for restoring consistency.

\paragraph{Variants of AICs.}
The authors of~\cite{Flesca2004} also considered \emph{conditioned active integrity constraints}, where the
actions on the head of AICs are guarded by additional conditions that have to be satisfied.
In their setting, conditioned AICs do not add expressive power to the formalism, as they can be split into
several unconditioned AICs (with more specific bodies) preserving the notions of consistency and repairs.
In our setting, this transformation is not possible, and it would thus be interesting to study conditioned
active integrity constraints over multi-context systems.
However, we point out that the management function \emph{can} use information about the actual knowledge bases
in its implementation, so some conditions can actually be expressed in our setting (see
Example~\ref{ex:min-card}).

\paragraph{Conclusion.}
We proposed active integrity constraints for multi-context systems and showed that, using them, we can compute
grounded repairs for inconsistent MCSs automatically.
Although validity of AICs is in general undecidable, we showed that we can cover the most common types of ICs
in our framework.

\bibliographystyle{plain}
\bibliography{bibl}

\end{document}